\documentclass[12pt,a4paper]{article}
\usepackage{amsfonts,amsmath,amssymb,fullpage,graphicx,mathrsfs}
\usepackage{algorithm,algorithmic,bm,enumerate}

\newtheorem{corollary}{Corollary}[section]
\newtheorem{definition}[corollary]{Definition}
\newtheorem{example}[corollary]{Example}
\newtheorem{lemma}[corollary]{Lemma}

\newtheorem{proposition}[corollary]{Proposition}

\newtheorem{theorem}[corollary]{Theorem}

\newcommand{\pmat}[1]{\begin{pmatrix}#1\end{pmatrix}}

\newcommand{\remove}[1]{}
\newcommand{\qed}{\rule{.07in}{.1in}}
\newenvironment{proof}{\vspace{1ex}\noindent\textbf{Proof}\hspace{0.5em}}{\hfill\qed\vspace{1ex}}

\title{Multidimensional Fibonacci Coding}
\author{Perathorn Pooksombat, Patanee Udomkavanich, and Wittawat Kositwattanarerk
\thanks{P. Pooksombat and W. Kositwattanarerk are with the Department of Mathematics, Faculty of Science, Mahidol University, Bangkok 10400, Thailand and the Centre of Excellence in Mathematics, the Commission on Higher Education, Bangkok 10400, Thailand (e-mail: perathorn.pok@student.mahidol.edu and wittawat.kos@mahidol.edu). P. Udomkavanich is with the Department of Mathematics and Computer Science, Faculty of Science, Chulalongkorn University, Bangkok 10330, Thailand (email: pattanee.u@chula.ac.th). The research of W. Kositwattanarerk for this work is supported by the Thailand Research Fund under Research Grant MRG6180192.}}
\date{\today}

\begin{document}

\maketitle

\begin{abstract}
Fibonacci codes are self-synchronizing variable-length codes that are proven useful for their robustness and compression capability. Asymptotically, these codes provide better compression efficiency as the order of the underlying Fibonacci sequence increases, but at the price of the increased suffix length. We propose a circumvention to this problem by introducing higher-dimensional Fibonacci codes for integer vectors. In the process, we provide extensive theoretical background and generalize the theorem of Zeckendorf to higher order. As thus, our work unify several variations of Zeckendorf's theorem while also providing new grounds for its legitimacy.

\end{abstract}


\section{Introduction}

A fundamental idea in data compression is to assign shorter codes to symbols or groups of symbols that appear more often in the source data. This results in what is called variable-length code, and one of the earliest practical compression code is the Huffman codes \cite{H}. This code approaches theoretical compression limit under certain scenarios, and has recently been amended to create Tagged Huffman Codes \cite{MNZB} and End-Tagged Dense Codes \cite{BINP} to accommodate database search. \\

Popular alternatives to these codes are codes with suffix delimiters, with the most well-known ones being Fibonacci codes. The scheme was introduced by \cite{FK} and generalized to higher order in \cite{AF}. In terms of data compression, Fibonacci codes are optimal under some distributions and can be used as an alternative to Huffman codes \cite{FK96,S}. Fibonacci coding is a result of binary numeration system for the integers known as Zeckendorf representation. This makes Fibonacci codes suitable for integer data, and equips them with fast encoding and decoding algorithms \cite{WKP,WKBPS}. In addition, compared with other variable-length codes for the integers such as logarithmic ramp \cite{ER} and Elias codes \cite{E}, Fibonacci codes provide much better resistance against insertion and deletion errors. Other applications of Fibonacci sequence in information science include the study of Morse code as a monoid \cite{C,Sa} and Wavelet trees \cite{KlS}. \\

While suffix codes have been shown to offer adjustability to some extent \cite{AZ}, a great burden for codes with suffix delimiters to bear is its own suffix. Longer suffix gives more flexibility in the coded data, hence yielding denser codes, but obviously the codewords now have to carry longer appendage. In this paper, we introduce Fibonacci coding for a sequential string of integers. In the proposed scheme, a constituent of integers is encoded together with a single delimiter, thus making Fibonacci codes of higher order more practical. \\

To generalize Fibonacci codes to the full extent, we reexamine their origin, which is the theorem of Zeckendorf. This much-celebrated theorem states that every positive integer can be written uniquely as
\[F_{i_1}+F_{i_2}+\ldots+F_{i_j}\]
where $F_i$'s are Fibonacci numbers and $i_1,i_2-i_1,\ldots,i_j-i_{j-1}\geq2$. In other words, one may represent a positive integer as sum of nonconsecutive Fibonacci numbers. This theorem has been generalized and investigated in many directions; similar results can be stated for negaFibonacci sequence \cite{B,K}, generalized Fibonacci sequence of order $k$ \cite{CHS,U}, and linear recurrence sequences \cite{DDKMMV,GTNP}. The distribution of the number of terms in a Zeckendorf decomposition is studied in \cite{BDEMMTW,L,MW} where a probabilistic approach has recently been applied \cite{BaM}. \\

%Zeckendorf's theorem also finds several applications in data transmission and compression. Since the theorem provides a binary numeration system for the integers, it can be used to design error-resistant universal codes called Fibonacci codes. Theoretically, they make use of the fact that the Zeckendorf representation of an integer involves no two consecutive terms from the Fibonacci sequence.
%Since the digits used are binary, Fibonacci coding provides a decent coding scheme for communications over a binary channel.

Noninteger Fibonacci sequences have been studied mostly on complex numbers \cite{Be,G}. In particular, Jordan \cite{J} and Harman \cite{Ha} consider Fibonacci sequences whose initial terms and indices are Gaussian integers. Other generalizations usually involve an identity whose integer terms yield the classical Fibonacci numbers. For example, Binet's formula is exploited in \cite{Hal,P} where a Fibonacci function is defined over the real and complex numbers. \\

%From number-theoretic point of view, Zeckendorf decomposition can be seen as a numeration system where binary digits are used in conjunction with Fibonacci numbers as a base. As it turns out, binary representation for multidimensional number systems is not well studied. K\'atai and Szab\'o \cite{KS} show that Gaussian integers can be written in base $-1+i$, thus giving a binary representation for Gaussian integers where the real and complex part are not dealt with separately. In \cite{F1}, Forney uses $1+i$ as a basis to construct complex lattice codes. While some ring of integers may have special characteristics that permit similar binary representation, these systems at best provide fixed-length codes for communication over a channel \cite{PHK}. \\

In this paper, we are interested in a generalization of Zeckendorf's theorem in its purest form: we consider \textit{any} mathematical sequences that satisfy the generalized Fibonacci recurrence relation and examine elements that can be written as a sum of terms from the sequence.
%Since these sequences are defined by their initial conditions, elements of interest are practically integral linear combination of the sequence's initial conditions.
As a result, mathematical spaces with minimal structure that meet our requirement are the free $\mathbb{Z}$-modules. Loosely speaking, they are spaces of vectors with integer components. This choice of spaces allows our generalizations to encompass, for example, integer vectors, arrays, or blocks of integers of fixed length; Gaussian and Eisenstein integers; lattices; and the ring $\mathbb{Z}[\alpha]$ where $\alpha$ is algebraic. \\

%In fact, note that Zeckendorf theorem also hold for Lucas numbers \cite{Br}.Notice that we do not specify the initial conditions for a Fibonacci sequence. In fact, most of our results hold irrespective of the initial conditions of the sequence.

%This paper studies generalizations of Zeckendorf's theorem for modules and the subsequent Fibonacci coding.
We summarize the contributions of this paper as follows.
\begin{itemize}
\item We introduce the notion of $k$-equivalent sequences. This approach technically treats elements that can be generated from a Fibonacci sequence as a number system using sequence elements as a basis. This key mechanism allows us to manipulate the ``digits'' without an explicit knowledge of the underlying Fibonacci sequence. We prove that an element has a Zeckendorf decomposition if and only if it can be written as a finite sum (with multiplicities) of sequence elements.
%``flatten out'' coefficients; roughly speaking, any integer combinations of elements from a Fibonacci sequence can be written as a sum of distinct elements from the same sequence.
\item We provide a sufficient and necessary condition for Zeckendorf's theorem to hold in a free $\mathbb{Z}$-module. Obviously, the classical Zeckendorf's theorem can be viewed as a case of our generalization. We also give a sufficient condition that makes the representation unique.
\item We propose Fibonacci coding for free $\mathbb{Z}$-modules. This is the first work that makes possible native Fibonacci coding for non-integers. Not only that the resulting codes inherit robustness property from the standard Fibonacci codes, they excel in terms of compression efficiency. Natural and explicit encoding and decoding algorithm are also given.
%our work substantially enlarges the message spaces for Fibonacci code,
%Naturally, this also gives binary numeration systems for spaces such as Gaussian integers, Eisenstein integers, quarternions, the ring of integers of any algebraic number field, and lattices.
\end{itemize}

The remainder of this paper is organized as follows. Section \ref{sec:def} sets the definitions that will be used throughout the paper. In Section \ref{sec:equi}, we prove several useful results concerning $k$-equivalent sequences. Generalizations of Zeckendorf's theorem are given in Section \ref{sec:main}. We establish multidimensional Fibonacci coding and discuss its efficiency in Section \ref{sec:code} and conclude in Section \ref{sec:conclusion}.

\section{Definitions}\label{sec:def}

We first give a definition for free modules and Fibonacci sequences of higher order. A free $\mathbb{Z}$-module is a module over $\mathbb{Z}$ with a basis. Namely, $M$ is a free $\mathbb{Z}$-module of rank $l$ if it is a group under addition and
\[M=\alpha_1\mathbb{Z}\oplus\alpha_2\mathbb{Z}\oplus\ldots\oplus\alpha_l\mathbb{Z}\]
for some $\alpha_1,\alpha_2,\ldots,\alpha_l$ that are integrally independent, meaning that the only integer solution to the equation $n_1\alpha_1+n_2\alpha_2+\ldots+n_l\alpha_l=0$ is $n_1=n_2=\ldots=n_l=0$. The elements $\alpha_1,\alpha_2,\ldots,\alpha_l$ are called a basis for $M$, and every element of $M$ can be written uniquely as $n_1\alpha_1+n_2\alpha_2+\ldots+n_l\alpha_l$ where $n_1,n_2,\ldots,n_l\in\mathbb{Z}$. For $n\in\mathbb{Z}^+$ and $m\in M$, we write $nm$ to mean $\underbrace{m+m+\ldots+m}_{n\textrm{ times}}$. Note that multiplication between module elements may or may not to be defined. \\

% If $\alpha_1,\alpha_2,\ldots,\alpha_l$ are integrally independent, $\mathbb{Z}\langle\alpha_1,\alpha_2,\ldots,\alpha_l\rangle$ is the free $\mathbb{Z}$-module generated by $\alpha_1,\alpha_2,\ldots,\alpha_l$. For example, $\mathbb{Z}\langle1\rangle$, $\mathbb{Z}\langle(1,0,0),(0,1,0),(0,0,1)\rangle$, $\mathbb{Z}\langle1,i\rangle$, $\mathbb{Z}\langle1,\omega\rangle$ are the set of integers, $\mathbb{Z}^3$, Gaussian integers, and Eisenstein integers respectively. \\

%We say that $\mathbf{v}_1,\mathbf{v}_2,\ldots,\mathbf{v}_l\in\mathbb{Z}^l$ are a basis for $\mathbb{Z}^l$ if every element in $\mathbb{Z}^l$ is an integer linear combination of $\mathbb{Z}^l$.

We are interested in an extremely broad version of a Fibonacci sequence, so we adopt only the fundamental recurrence relation and a corresponding generalization of a Zeckendorf representation.

\begin{definition}
A \textit{Fibonacci sequence of order $k$} is a doubly infinite sequence $(F_r)_{r\in\mathbb{Z}}=(\ldots,F_{-1},F_0,F_1,\ldots)$ where
\[F_{n-k}+\ldots+F_{n-2}+F_{n-1}=F_n\]
for all integer $n$.
\end{definition}

Here, a doubly infinite sequence (also called two-way infinite sequence) is a sequence that goes in both directions.

\begin{definition}
Let $(F_r)_{r\in\mathbb{Z}}$ be a Fibonacci sequence of order $k$. An element $m$ is \textit{Zeckendorf in $(F_r)$} if it can be written as a finite sum of elements in $(F_r)$ with no $k$ consecutive terms. A set $M$ is Zeckendorf in $(F_r)$ if every element of $M$ is.
\end{definition}

%While our definition permits a Fibonacci sequence of practically any mathematical objects, we will focus on a sequence of integer vectors.
%Here, note that the entries of our sequences are allowed to be vectors, and we do not specify the initial conditions for a Fibonacci sequence. In fact, we will prove that any initial condition is fine as long as it spans the space of integer vectors.
%In other words, given a Fibonacci sequence whose elements are vectors from $\mathbb{Z}^l$, exactly when does every vector in $\mathbb{Z}^l$ has a Zeckendorf representation?

Note that we do not specify the initial condition for a Fibonacci sequence. In fact, we will reverse engineer the theorem of Zeckendorf and identify all initial conditions that the theorem holds. In other words, instead of proving a version of Zeckendorf's theorem for a particular sequence, we set forth to classify \textit{all} sequences that are in favor of the theorem. \\

A Fibonacci sequence of order $k$ has degree of freedom $k$; knowing any $k$ consecutive terms is sufficient to generate the entire sequence. For $k\geq2$, the characteristic polynomial of a Fibonacci sequence of order $k$ is $x^k-x^{k-1}-x^{k-2}-\ldots-x-1$. It is known \cite{M} that this polynomial has $k$ distinct roots $\lambda_k,\lambda_{k,2},\ldots,\lambda_{k,k}$ where $\lambda_k\in(1,2)$ and $\lambda_{k,2},\ldots,\lambda_{k,k}$ lie in the unit circle. Note that the sequence $(\lambda_k^r)_{r\in\mathbb{Z}}$ is a Fibonacci sequence of order $k$. We call this sequence \textit{primitive}.

\begin{example}\label{ex:basic}
In this example we consider several Fibonacci sequences of order 3.
\begin{enumerate}[i.)]
\item The usual Tribonacci sequence is defined using $F_0=F_1=1$, $F_2=2$, and $F_{n-3}+F_{n-2}+F_{n-1}=F_n$. The first few terms of this sequence are
\[\begin{array}{rrrrrrrrrrrr}
1, & 1, & 2, & 4, & 7, & 13, & 24, & 44, & 81, & 149, & 274, & \ldots.
\end{array}\]
The two-way generalization of this sequence is given by
\[\begin{array}{rrrrrrrrrrrr}
\ldots, & 5, & -8, & 4, & 1, & -3, & 2, & 0, & -1, & 1, & 0, & 0, \\
\end{array}\]
\[\begin{array}{rrrrrrrrrrrr}
1, & 1, & 2, & 4, & 7, & 13, & 24, & 44, & 81, & 149, & 274, & \ldots.
\end{array}\]
\item The three roots of a polynomial $x^3-x^2-x-1$ are $\lambda_3\approx1.839$, $\lambda_{3,2}\approx-0.42-0.606i$, and $\lambda_{3,3}\approx-0.42+0.606i$. The primitive Fibonacci sequence of order 3 is the sequence $(\lambda_3^r)_{r\in\mathbb{Z}}$, which is approximately
\[\begin{array}{rrrrrrrrrrr}
\ldots, & 0.161, & 0.296, & 0.544, & 1, & 1.839, & 3.383, & 6.222, & 11.445, & 21.05, & \ldots.
\end{array}\]
\item A Fibonacci sequence of order 3 where $F_0=0$, $F_1=1+i$, and $F_2=2+i$ is given by
\[\begin{array}{rrrrrrrrr}
\ldots, & 5+i, & -2-3i, & -1+2i, & 2, & -1-i, & i, & 1, & \\
& 0, & 1+i, & 2+i, & 3+2i, & 6+4i, & 11+7i, & 20+13i, & \ldots.
\end{array}\]
Here, the element of this sequence are Gaussian integers, i.e., complex numbers of the form $a+bi$ where $a,b\in\mathbb{Z}$. The element $9+5i$, for example, is Zeckendorf in this sequence since $9+5i=(5+i)+(-1+2i)+(2)+(3+2i)$.
\end{enumerate}
\end{example}

In order to study elements that are Zeckendorf in $(F_r)$, it is natural to consider elements of the form $(x_r)\cdot(F_r)$ where $(x_r)$ is a doubly infinite sequence of integers and $\cdot$ is an infinite dot product, i.e.
\[(x_r)\cdot(F_r)=\sum_{r\in\mathbb{Z}}{x_rF_r}.\]
Roughly speaking, a collection of all elements of the form $(x_r)\cdot(F_r)$ is a number system with the base $(F_r)$ with the digit set $\mathbb{Z}$. Here, $(x_r)$ acts like digits or coefficients for the basis $(F_r)$. These so-called coefficients can be manipulated using the following methods.

\begin{definition}\label{def:eq}
Doubly infinite sequences $(x_r)$ and $(y_r)$ are \textit{$k$-equivalent}, denoted $(x_r)\sim_k(y_r)$, if $(y_r)$ can be obtained from $(x_r)$ using finite applications of the following two operations. \\
\textit{Operation 1}: For some $n\in\mathbb{Z}$, subtract 1 from $x_n$ and add 1 to all of $x_{n-k}\ldots,x_{n-2},x_{n-1}$. \\
\textit{Operation 2}: For some $n\in\mathbb{Z}$, add 1 to $x_n$ and subtract 1 from all of $x_{n-k}\ldots,x_{n-2},x_{n-1}$. \\
\end{definition}

It is not hard to see that \textit{Operations 1} and \textit{2} ``cancel out'', and $\sim_k$ is an equivalence relation. These operations also preserve the value of the dot product: if $(F_r)$ is a Fibonacci sequence of order $k$ and $(x_r)\sim_k(y_r)$, then $(x_r)\cdot(F_r)=(y_r)\cdot(F_r)$. As we will see later, this observation is an important ingredient of our results. Finally, we give a lexicographical order $\succ$ to the doubly infinite sequences that are zero almost everywhere, i.e., $(x_r)_{r\in\mathbb{Z}}\succ(y_r)_{r\in\mathbb{Z}}$ if the entry of $(x_r)$ is larger than that of $(y_r)$ at the rightmost index where they are not equal.

\section{$k$-equivalent Sequences}\label{sec:equi}

In a way, our approach in proving a generalized theorem of Zeckendorf is in asking the following converse question: given a Fibonacci sequence $(F_r)$ of order $k$, which elements can be written as a finite sum of elements from $(F_r)$ using no $k$ consecutive terms? It is not hard to see that they are elements of the form $(y_r)\cdot(F_r)$ where $(y_r)$ is binary, is 0 almost everywhere, and has no $k$ consecutive 1's. Now, instead of characterizing these elements directly, we will focus on manipulating the coefficients $(y_r)$ using \textit{Operations 1} and \textit{2} from Definition \ref{def:eq}. These coefficient manipulations operate independently of the underlying Fibonacci sequence $(F_r)$, and so this approach frees us from having to worry about the choice of $(F_r)$. As we will see, our results hold in a broad sense and are not limited to just any specific Fibonacci sequence. Propositions \ref{prop:keq1} and \ref{prop:keq2} given in this section will provide an effective mechanism for identifying sequences that are $k$-equivalent to a binary sequence with the required properties.

\begin{proposition}\label{prop:keq1}
For any sequence $(x_r)_{r\in\mathbb{Z}}$ of integers that is 0 almost everywhere, there exists a sequence $(y_r)_{r\in\mathbb{Z}}$ such that $(x_r)\sim_k(y_r)$, and $(y_r)_{r\in\mathbb{Z}}$ is either
\begin{itemize}
\item zero everywhere,
\item positive at some $k$ consecutive terms and zero elsewhere, or
\item negative at some $k$ consecutive terms and zero elsewhere.
\end{itemize}
In addition, if $(x_r)_{r\in\mathbb{Z}}\neq(0)_{r\in\mathbb{Z}}$ is nonnegative, then the $k$ consecutive nonzero terms of $(y_r)_{r\in\mathbb{Z}}$ are positive.
\end{proposition}

\begin{proof}
We keep applying either \textit{Operation 1} or \textit{2} from Definition \ref{def:eq} to eliminate the rightmost nonzero term of $(x_r)$. Since $(x_r)$ is 0 almost everywhere, we will eventually obtain a sequence $(y_r)_{r\in\mathbb{Z}}$ which is 0 everywhere except for some $k$ consecutive terms. If the $k$ consecutive terms of $(y_r)$ are all positive, all negative, or all zero, then we are done. Otherwise, we can keep eliminating the rightmost nonzero term of $(y_r)$ and obtain yet another sequence that is 0 everywhere except for some $k$ consecutive terms. We will prove that this process will eventually yield a sequence with $k$ consecutive terms that are either all positive or all negative. \\

We initialize $(y_{0,r})=(y_r)$ and collect the $k$ nonzero terms of $(y_{0,r})$ into a $1\times k$ vector $\boldsymbol{\gamma}_0=(\gamma_1,\gamma_2,\ldots,\gamma_k)$. The rightmost term of $(y_{0,r})$ is $\gamma_k$. Using either \textit{Operation 1} or \textit{2}, we eliminate this term an obtain $(y_{1,r})$ whose nonzero entries are given by $\boldsymbol{\gamma}_1=(\gamma_k,\gamma_1+\gamma_k,\gamma_2+\gamma_k,\ldots,\gamma_{k-1}+\gamma_k)$. In other words, we have
\[\boldsymbol{\gamma}_1^\top=A\boldsymbol{\gamma}_0^\top\]
where
\[A=\pmat{
0 & 0 & 0 & \cdots & 0 & 1 \\
1 & 0 & 0 & \cdots & 0 & 1 \\
0 & 1 & 0 & \cdots & 0 & 1 \\
\vdots & \vdots & \vdots & \ddots & \vdots & \vdots \\
0 & 0 & 0 & \cdots & 1 & 1}.\]
Note that this is independent of the value of $\gamma_k$. We repeat this process, and it is not hard to see that $(y_{0,r})$ is $k$-equivalent to $(y_{n,r})$ for all positive integer $n$. In addition, the nonzero terms of $(y_{n,r})$ are given by the entries of $\boldsymbol{\gamma}_n^\top=A^n\boldsymbol{\gamma}_0^\top$. Thus, it remains to show that there is an $n$ for which $A^n\boldsymbol{\gamma}_0^\top$ is either positive or negative. \\

It should not be surprising that $A$ is the transpose of the usual Fibonacci matrix. Thus, $A$ has characteristic equation $x^k-x^{k-1}-x^{k-2}-\ldots-x-1$ and eigenvalues $\lambda_k\in(1,2)$ and $\lambda_{k,2},\ldots,\lambda_{k,k}$ whose complex norms are less than 1 \cite{M}. We denote by $\bm{v}_i$ eigenvector corresponding to the eigenvalue $\lambda_i$. Note that
\[\bm{v}_1=\pmat{\lambda_k^{k-2} \\
\lambda_k^{k-2}+\lambda_k^{k-3} \\
\lambda_k^{k-2}+\lambda_k^{k-3}+\lambda_k^{k-4} \\
\vdots \\
\lambda_k^{k-2}+\lambda_k^{k-3}+\cdots+\lambda_k+1 \\
\lambda_k^{k-1}}\]
is real and positive. Now, one may diagonalize $A$ as $CDC^{-1}$ where $D=\mathsf{diag}(\lambda_k,\lambda_{k,2},\ldots,\lambda_{k,k})$ and $C=\pmat{\bm{v}_1 & \cdots & \bm{v}_k}$. It follows that
\begin{align*}
\boldsymbol{\gamma}_n^\top & =A^n\boldsymbol{\gamma}_0^\top \\
& =CD^nC^{-1}\boldsymbol{\gamma}_0^\top \\
& =\pmat{\lambda_k^n\bm{v}_1 & \lambda_{k,2}^n\bm{v}_1 & \cdots & \lambda_{k,k}^n\bm{v}_k}C^{-1}\boldsymbol{\gamma}_0^\top \\
& =\lambda_k^nc_1\bm{v}_1+\lambda_{k,2}^nc_2\bm{v}_2+\ldots+\lambda_{k,k}^nc_k\bm{v}_k
\end{align*}
where $(c_1,\ldots,c_k)^\top=C^{-1}\boldsymbol{\gamma}_0^\top$. Since $|\lambda_k|>1$, $|\lambda_{k,2}|,\ldots,|\lambda_{k,k}|<1$, and $\bm{v}_1$ is real and positive, for a sufficiently large $n$ we will have $\boldsymbol{\gamma}_n^\top\approx\lambda_k^nc_1\bm{v}_1$, and so the entries of $\boldsymbol{\gamma}_n^\top$ will either be all positive or all negative, depending on the sign of $c_1$. If $c_1=0$, then $\boldsymbol{\gamma}_n^\top=\mathbf{0}$, implying that $\boldsymbol{\gamma}_0^\top=\mathbf{0}$. \\

For the last part of the proposition, if $(x_r)_{r\in\mathbb{Z}}\neq(0)_{r\in\mathbb{Z}}$ is nonnegative, then we only need to use \textit{Operation 1}, and so the $k$ consecutive terms of $(y_r)_{r\in\mathbb{Z}}$ cannot be negative or zero.
\end{proof}

\begin{proposition}\label{prop:keq2}
For any sequence $(x_r)_{r\in\mathbb{Z}}$ of nonnegative integers that is 0 almost everywhere, there exists a sequence $(y_r)_{r\in\mathbb{Z}}$ of 0 and 1 with no $k$ consecutive 1's such that $(x_r)\sim_k(y_r)$.
\end{proposition}

\begin{proof}
Consider all sequences of nonnegative integers that are $k$-equivalent to $(x_r)_{r\in\mathbb{Z}}$. Recall that if $(x_r)\sim_k(y_r)$, then $(x_r)\cdot(\lambda_k^r)=(y_r)\cdot(\lambda_k^r)$ where $(\lambda_k^r)$ is the primitive Fibonacci sequence. Let $N$ be an integer such that $(x_r)\cdot(\lambda_k^r)<\lambda_k^N$. Now, if a sequence of nonnegative integers $(y_r)$ is $k$-equivalent to $(x_r)$, then we must have $(y_r)\cdot(\lambda_k^r)<\lambda_k^N$. Since $(\lambda_k^r)$ is strictly positive, it follows that $y_r=0$ for all $r\geq N$. This allows us to conclude that, among all the sequences of nonnegative integers that are $k$-equivalent to $(x_r)$, there must be one with the highest lexicographical order. We call this sequence $(y_r)_{r\in\mathbb{Z}}$ and claim that it satisfies the conditions required. \\

Suppose that $y_s\geq2$ for some $s\in\mathbb{Z}$. We perform \textit{Operation 1} at $n=s$ and \textit{Operation 2} at $n=s+1$. This results in a sequence with higher lexicographical order than $(y_r)$, contradicting our choice of $(y_r)$. Suppose now that $y_{s-k}=\ldots=y_{s-2}=y_{s-1}=1$ for some $s\in\mathbb{Z}$. Perform \textit{Operation 2} at $n=s$ yields a sequence with higher lexicographical order, and once again that contradicts the choice of $(y_r)$. We conclude that $(y_r)$ consists only of 0 and 1 with no $k$ consecutive 1's.
\end{proof}

We illustrate the transformations given in Propositions \ref{prop:keq1} and \ref{prop:keq2} in an example below.

\begin{example}
We let $k=3$ and consider a doubly infinite sequence $(x_r)$ that is zero everywhere except $x_{-1}=-2$, $x_1=-1$, $x_2=-2$, and $x_4=1$, i.e.,
\[(x_r)=(\ldots,0,0,0,0,-2,0,-1,-2,0,1,\ldots).\]
Here and throughout the example, only indices $-5,-4,\ldots,4$ are shown. We now transform $(x_r)$ into a sequence that is positive at some 3 consecutive terms and zero elsewhere.
\begin{align*}
(\ldots,0,0,0,0,-2,0,-1,-2,0,1,\ldots) & \sim_3(\ldots,0,0,0,0,-2,0,0,-1,1,0,\ldots) \\
& \sim_3(\ldots,0,0,0,0,-2,1,1,0,0,0,\ldots) \\
& \sim_3(\ldots,0,0,0,1,-1,2,0,0,0,0,\ldots) \\
& \sim_3(\ldots,0,0,1,2,0,1,0,0,0,0,\ldots). \\
& \sim_3(\ldots,0,0,2,3,1,0,0,0,0,0,\ldots).
\end{align*}
Here, \textit{Operation 1} is performed at indices 4,3,1,0,0 in order. Next, we transform $(\ldots,0,0,2,3,1,0,0,0,0,0,\ldots)$ into a binary sequence with no 3 consecutive 1's.
\begin{align*}
(\ldots,0,0,2,3,1,0,0,0,0,0,\ldots) & \sim_3(\ldots,0,0,1,2,0,1,0,0,0,0,\ldots) \\
& \sim_3(\ldots,1,1,2,1,0,1,0,0,0,0,\ldots) \\
& \sim_3(\ldots,1,0,1,0,1,1,0,0,0,0,\ldots).
\end{align*}
Here, \textit{Operations 2,1,2} are performed in order at indices $1,-2,-1$ respectively. Consider now a Fibonacci sequence of Gaussian integers
\[(F_r)=(
\begin{array}{rrrrrrrrrrrr}
\ldots, & -1+2i, & 2, & -1-i, & i, & 1, & 0, & 1+i, & 2+i, & 3+2i, & 6+4i, & \ldots.
\end{array})\]
from Example \ref{ex:basic}. Note that all the above 3-equivalent sequences represent the same quantity when multiplied by $(F_r)$. In particular, if $(y_r)=(\ldots,0,0,2,3,1,0,0,0,0,0,\ldots)$ and $(z_r)=(\ldots,1,0,1,0,1,1,0,0,0,0,\ldots)$, then
\[(x_r)\cdot(F_r)=(y_r)\cdot(F_r)=(z_r)\cdot(F_r)=-1+i.\]
This, in fact, would hold for any other Fibonacci sequences $(F_r)$ of order 3.
\end{example}

The last example illustrates a major strength of our approach--each equivalency class represents the same quantity under a given Fibonacci sequence. If one can find a legitimate Zeckendorf representation in an equivalency class, then the element represented by that class has a Zeckendorf decomposition. We finish this section with a powerful corollary to Proposition \ref{prop:keq2} and another example. The next result characterizes all elements that are Zeckendorf in a Fibonacci sequence $(F_r)$.

\begin{corollary}\label{cor:rep}
Let $(F_r)$ be a Fibonacci sequence of order $k$. An element $m$ is Zeckendorf in $(F_r)$ if and only if $m$ can be written as a finite sum (with multiplicities) of elements from $(F_r)$.
\end{corollary}

\begin{proof}
An element $m$ is a finite sum of elements from $(F_r)$ if it is Zeckendorf in $(F_r)$. The converse of this statement follows from Proposition \ref{prop:keq2}.
\end{proof}

\begin{example}
Consider the primitive Fibonacci sequence of order 2,
\[\begin{array}{rrrrrrrrr}
\ldots, & \varphi^{-3}, & \varphi^{-2}, & \varphi^{-1}, & 1, & \varphi, & \varphi^2, & \varphi^3, & \ldots,
\end{array}\]
where $\varphi=\frac{1+\sqrt{5}}{2}$ is the golden ratio. The elements that are Zeckendorf in this sequence are precisely positive elements in the ring $\mathbb{Z}[\varphi]$, and this numeration system is known as golden ratio base.
\end{example}

\section{Zeckendorf's Theorem for Free $\mathbb{Z}$-Modules}\label{sec:main}

Intuitively, if an element $m$ is Zeckendorf in a Fibonacci sequence $(F_r)$ of order $k$, then it is an integer combination of the initial terms (or, in fact, any $k$ terms) of $(F_r)$. Integer span of these forms a module, so it is natural to attempt to represent every element of a module using sequence elements. Subsection \ref{subsec:2way} below deals with doubly infinite Fibonacci sequences using tools developed in the previous section. The result is then specialized to one-way Fibonacci sequences in Subsection \ref{subsec:unique}. This paves the way for the Zeckendorf decomposition and Fibonacci coding for modules.

\subsection{Two-sided Sequences}\label{subsec:2way}

Corollary \ref{cor:rep} lays a strong foundation for our work. While it characterizes elements that are Zeckendorf in a sequence, it gives insufficient information on the algebraic structure of the set of all elements that have a legitimate representation. The following theorem now reverses the process. It gives a sufficient and necessary condition for every element of a free $\mathbb{Z}$-module $M$ to be Zeckendorf in a Fibonacci sequence $(F_r)$.

\begin{theorem}\label{th:main}
Let $M$ be a free $\mathbb{Z}$-module of rank $l$, and $(F_r)_{r\in\mathbb{Z}}=(\ldots,F_{-1},F_0,F_1,\ldots)$ be a Fibonacci sequence of order $k$. Then, $M$ is Zeckendorf in $(F_r)$ if and only if both of the following conditions are satisfied.
\begin{enumerate}
\item $F_i,\ldots,F_{i+k-1}$ span $M$ for all $i$.
\item $l<k$.
\end{enumerate}
\end{theorem}

\begin{proof}
We first remark that the integral span of $F_i,\ldots,F_{i+k-1}$ is the same as the integral span of $F_{i+1},\ldots,F_{i+k}$ since $F_i+\cdots+F_{i+k-1}=F_{i+k}$. Thus, $F_{i},\ldots,F_{i+k-1}$ span $M$ for \textit{all} $i$ if and only if $F_{i},\ldots,F_{i+k-1}$ span $M$ for \textit{any} $i$. \\

It is easy to see that the first condition is necessary. Since the elements of $(F_r)_{r\in\mathbb{Z}}$ are integer combinations of $F_i,\ldots,F_{i+k-1}$, so must be any sums of the elements from this sequence. It follows that $F_i,\ldots,F_{i+k-1}$ span $M$. This readily implies $l\leq k$. We will now show that $l$ cannot equal to $k$, thus establishing the necessity of the second condition. \\

Suppose that $l=k$, and let $m$ be any nonzero element of $M$. Since both $m$ and $-m$ are Zeckendorf in $(F_r)$, we have
\[m=(x_r)\cdot(F_r)\quad\textrm{ and }\quad-m=(y_r)\cdot(F_r)\]
for some sequences $(x_r),(y_r)$ of 0 and 1 with no $k$ consecutive 1's. Hence, $0=(x_r+y_r)\cdot(F_r)$. By Proposition \ref{prop:keq1}, $(x_r+y_r)$ is $k$-equivalent to a sequence that is positive at some $k$ consecutive terms and zero elsewhere. Thus, we have
\[0=(x_r+y_r)\cdot(F_r)=(z_r)\cdot(F_r)=z_iF_i+z_{i+1}F_{i+1}+\ldots+z_{i+k-1}F_{i+k-1}\]
for some $i$ in which $z_i,z_{i+1},\ldots,z_{i+k-1}>0$. This, however, means that $F_i,F_{i+1},\ldots,F_{i+k-1}$ are integrally dependent, and so they cannot span a module of rank $l=k$. \\

We are left to show that the two conditions are sufficient. Since $F_0,\ldots,F_{k-1}$ span $M$ and $l+1\leq k$, $F_1,\ldots,F_{k-1}$ must be integrally dependent. That is, we may write
\[0=x_0F_0+\cdots+x_{k-1}F_{k-1}\]
where $x_0,\ldots,x_{k-1}$ are not all zeros. In other words, we have
\[0=(x_r)\cdot(F_1)\]
where $(x_r)$ is zero everywhere but some $k$ consecutive terms. We apply Proposition \ref{prop:keq1} and obtain a sequence $(y_r)$ that is $k$-equivalent to $(x_r)$ and contain $k$ consecutive terms of the same sign and zero elsewhere. Note that $(y_r)$ cannot be zero everywhere since that would imply that $x_0=\ldots=x_{k-1}=0$. \\

Suppose now that $(y_r)$ is nonzero at the terms $y_i,\ldots,y_{i+k-1}$. This means
\[0=y_iF_i+\cdots+y_{i+k-1}F_{i+k-1}\]
where either $y_i,\ldots,y_{i+k-1}>0$ or $y_i,\ldots,y_{i+k-1}<0$. Since $F_i,\ldots,F_{i+k-1}$ span $M$, we can write any element $m\in M$ as
\[m=z_iF_i+\cdots+z_{i+k-1}F_{i+k-1}\]
where $z_i,\ldots,z_{i+k-1}$ are integers. Now, we have
\[m=m+0\cdot N=(z_i+y_iN)F_i+\cdots+(z_{i+k-1}+y_{i+k-1}N)F_{i+k-1}\]
for all integer $N$. For a sufficiently large (and possibly negative) $N$, the terms $z_i+y_iN,\ldots,z_{i+k-1}+y_{i+k-1}N$ will all be positive, and it follows from Corollary \ref{cor:rep} that $m$ is Zeckendorf in $(F_r)$. This completes the proof of the theorem.
\end{proof}

\begin{corollary}\label{cor:span}
Let $(F_r)_{r\in\mathbb{Z}}=(\ldots,F_{-1},F_0,F_1,\ldots)$ be a Fibonacci sequence of order $k$. If $F_1,\ldots,F_k$ are integrally dependent, then every element in the integral span of $F_1,\ldots,F_k$ is Zeckendorf in $(F_r)$.
\end{corollary}

In view of Theorem \ref{th:main}, every integer has a (classical) Zeckendorf representation since $\mathbb{Z}$ is a module over itself with rank 1, and the negaFibonacci sequence is of order 2 with $F_0=0$ and $F_1=1$. In fact, the same result would hold as long as $F_0$ and $F_1$ are relatively prime (and hence span $\mathbb{Z}$ integrally). This is technically the case for the Lucas sequence where $L_0=2$ and $L_1=1$. Next, we give an example for the case of Gaussian integers.

\begin{example}\label{ex:represent}
Consider a Fibonacci sequence
\[\begin{array}{rrrrrrrrrr}
\ldots, & -4+4i, & 5+i, & -2-3i, & -1+2i, & 2, & -1-i, & i, & 1, & \\
& 0, & 1+i, & 2+i, & 3+2i, & 6+4i, & 11+7i, & 20+13i, & 37+24i, & \ldots.
\end{array}\]
from Example \ref{ex:basic}. It follows from Theorem \ref{th:main} that every Gaussian integer can be written as a sum of elements from this sequence with no 3 consecutive terms. For instance,
\begin{align*}
-2 & = (-2-3i)+(-1+2i)+(i)+(1) & -2i & =(-2-3i)+(2)+(i) \\
-1 & = (-1-i)+(i) & -i & =(-1-i)+(1) \\
2 & =(2) & 2i & = (-1+2i)+(1) \\
3 & =(2)+(1) & 3i & = (-1+2i)+(i)+(1) \\
1+2i & = (-1+2i)+(2) & 1-2i & = (-2-3i)+(i)+(1) \\
-2+i & =(-1+2i)+(-1-i) & -2-i & = (-2-3i)+(-1+2i)+(1).
\end{align*}
One may observe that $(i)+(1)=(1+i)$, and so Zeckendorf representation for a Gaussian integer in this sequence is not unique. Nonetheless, we will see in the next subsection how certain adjustments can make unique representation possible. \\

On the contrary, no Fibonacci sequence of order 2 can generate all Gaussian integers. For example, $-1-i$ cannot be written as a sum of elements from the sequence
\[\begin{array}{rrrrrrrrr}
\ldots, & -21+13i, & 13-8i, & -8+5i, & 5-3i, & -3+2i, & 2-i, & -1+i, & \\
& 1, & i, & 1+i, & 1+2i, & 2+3i, & 3+5i, & 5+8i, & \ldots.
\end{array}\]
Note that Corollary \ref{cor:span} does not apply here since $1$ and $i$ are not integrally dependent. However, it follows from Corollary \ref{cor:rep} that every element of the form $a+bi$, $a,b\in\mathbb{Z}^+$, is Zeckendorf in this sequence.
\end{example}

\subsection{One-sided Sequences and Unique Representation}\label{subsec:unique}

For dense code, we need a one-to-one mapping. However, we see in Example \ref{ex:represent} that a two-sided sequence do not provide such luxury. Note that this is observed in the classical Fibonacci sequence as well. It is only when the negaFibonacci sequence is extracted from the two-sided infinite sequence
\[\ldots,34,-21,13,-8,5,-3,2,-1,1,0,1,1,2,3,5,8,13,21,34,\ldots\]
that every integer has a unique Zeckendorf representation. In this subsection, we establish analogous results for Fibonacci sequences of higher order. The key ingredient of this development is Corollary \ref{cor:unique}, which states that binary sequences with no $k$ consecutive 1's cannot be $k$-equivalent.

%So far we have left out one very important aspect of Zeckendorf's Theorem: the uniqueness property. It is known that every positive integer can be written uniquely as a sum of nonconsecutive elements from the Fibonacci sequence
%\[1,2,3,5,8,13,21,34,\ldots,\]
%and every integer can be written uniquely as a sum of nonconsecutive elements from the negaFibonacci sequence
%\[\ldots,34,-21,13,-8,5,-3,2,-1,1.\]
%However, the uniqueness property obviously no longer holds for the two-way sequence
%\[\ldots,34,-21,13,-8,5,-3,2,-1,1,0,1,1,2,3,5,8,13,21,34,\ldots\]
%since this sequence contains elements that are equal.

\begin{lemma}\label{lem:bound}
Let $a$ be an integer, and let $(\lambda_k^r)_{r\in\mathbb{Z}}$ be the primitive Fibonacci sequence of order $k$. Then,
\[\sum_{\substack{r<a \\ k\,\nmid\,a-r}}{\lambda_k^r}=\lambda_k^a.\]
\end{lemma}

\begin{proof}
We have
\begin{align*}
\sum_{\substack{r<a \\ k\nmid a-r}}{\lambda_k^r} & =\sum_{r<a}{\lambda_k^r}-\sum_{\substack{r<a \\ k\mid a-r}}{\lambda_k^r} \\
& = \frac{\lambda_k^a}{\lambda_k-1}-\frac{\lambda_k^a}{\lambda_k^k-1} \\
& = \lambda_k^a\left(\frac{(\lambda_k^{k-1}+\lambda_k^{k-2}+\ldots+1)-1}{\lambda_k^k-1}\right) \\
& = \lambda_k^a\left(\frac{\lambda_k^k-1}{\lambda_k^k-1}\right) \\
& =\lambda_k^a.
\end{align*}
\end{proof}

\begin{corollary}\label{cor:unique}
Let $(x_r)$ and $(y_r)$ be doubly infinite binary sequences with no $k$ consecutive 1's that are 0 almost everywhere.
\begin{enumerate}
\item If $(x_r)\succ_k(y_r)$, then $(x_r)\cdot(\lambda_k^r)>(y_r)\cdot(\lambda_k^r)$.
\item If $(x_r)\sim_k(y_r)$, then $(x_r)=(y_r)$.
\end{enumerate}
\end{corollary}

\begin{proof}
Suppose that $(x_r)\succ_k(y_r)$, and let $a$ be the largest index where $(x_r)$ and $(y_r)$ are not equal. It follows that $(x_a)=1$, $(y_a)=0$, and
\begin{align*}
(x_r)\cdot(\lambda_k^r) & \geq\lambda_k^a+\sum_{r>a}{x_r\lambda_k^r} \\
& =\sum_{\substack{r<a \\ k\nmid a-r}}{\lambda_k^r}+\sum_{r>a}{x_r\lambda_k^r} \\
& >\sum_{r\leq a}{y_r\lambda_k^r}+\sum_{r>a}{y_r\lambda_k^r}=(y_r)\cdot(\lambda_k^r).
\end{align*}
Now, if $(x_r)\neq(y_r)$, then we may assume without loss of generality that $(x_r)\succ_k(y_r)$. This means $(x_r)\cdot(\lambda_k^r)>(y_r)\cdot(\lambda_k^r)$, and so $(x_r)\nsim_k(y_r)$. This readily establishes part \textit{ii}.
\end{proof}

Typically, $(x_r)\sim_k(y_r)$ means that $(x_r)\cdot(F_r)=(y_r)\cdot(F_r)$, and so $(x_r)$ and $(y_r)$ are two representations of the same element. Corollary \ref{cor:unique} now allows us to identify Fibonacci sequences (or parts of) that permit unique Zeckendorf representation: if a sequence $(F_r)$ has the property that $(x_r)\cdot(F_r)=(y_r)\cdot(F_r)$ implies $(x_r)\sim_k(y_r)$, then any element that is Zeckendorf in $(F_r)$ will have a unique representation. Since we are interested in generating every element in a $\mathbb{Z}$-module, we look into generalizing negaFibonacci sequence, which is technically a one-sided Fibonacci sequence to the left whose first excluded term is 0. It turns out that this observation generalizes well to modules.

\begin{theorem}\label{th:unique}
Let $k\geq 2$ be an integer, and $M$ be a free $\mathbb{Z}$-module of rank $k-1$. Let $(F_r)_{r\in\mathbb{Z}}=(\ldots,F_{-1},F_0,F_1,\ldots)$ be a Fibonacci sequence of order $k$ where $F_0=0$ and $F_{-k+1},\ldots,F_{-1}$ span $M$. Then, every element $m\in M$ can be uniquely written as a finite sum of elements from the one-way sequence
\[\ldots,F_{-3},F_{-2},F_{-1}\]
with no $k$ consecutive terms.
\end{theorem}

\begin{proof}
We first prove the existence. Let $m\in M$, and write
\[m=x_{-k+1}F_{-k+1}+\ldots+x_{-2}F_{-2}+x_{-1}F_{-1}\]
where $x_{-k+1},\ldots,x_{-2},x_{-1}$ are integers. We extend $x_{-k+1},\ldots,x_{-2},x_{-1}$ to a sequence $(x_r)_{r\in\mathbb{Z}}$ that is zero everywhere except $r=-k+1,\ldots,-2,-1$. Now, consider all sequences of integers that are $k$-equivalent to $(x_r)_{r\in\mathbb{Z}}$ such that the terms with positive index are zero and the terms with negative index are nonnegative. Such a sequence exists since one may apply \textit{Operation 1} from Definition \ref{def:eq} to $(x_r)$ at $n=0$ for $\max\{|x_{-k+1}|,\ldots,|x_{-2}|,|x_{-1}|\}$ times. Denote by $(y_r)_{r\in\mathbb{Z}}$ the sequence that has the aforementioned properties with highest lexicographical order. We claim that $(y_r)$ gives a Zeckendorf representation for $m$ in $\ldots,F_{-3},F_{-2},F_{-1}$. \\

We first note that $y_r=0$ for $r\in\mathbb{Z}^+$ and $F_0=0$, and so $(y_r)\cdot(F_r)=\sum_{r\in\mathbb{Z}^-}{y_rF_r}$. If $y_s\geq2$ for some $s\in\mathbb{Z}^-$, then we perform \textit{Operation 1} at $n=s$ and \textit{Operation 2} at $n=s+1$ to obtain a sequence with higher lexicographical order than $(y_r)$. If $y_{s-k+1}=\ldots=y_{s-1}=y_{s}=1$ for some $s\in\mathbb{Z}^-$, then performing \textit{Operation 2} at $n=s$ results in a sequence with higher lexicographical order. Thus, $(y_r)$ is binary with no $k$ consecutive 1's, and $\sum_{r\in\mathbb{Z}^-}{y_rF_r}$ is a Zeckendorf representation for $m$ in $\ldots,F_{-3},F_{-2},F_{-1}$. \\

Suppose now for the sake of contradiction that there exists an element $m\in M$ that has two representations as a finite sum of elements from the sequence $\ldots,F_{-3},F_{-2},F_{-1}$ with no $k$ consecutive terms. That is, we have $m=(x_r)\cdot(F_r)=(y_r)\cdot(F_r)$ where $(x_r)\neq(y_r)$ are doubly infinite binary sequences which are zero for $r\geq0$. Now, we keep applying either \textit{Operation 1} or \textit{2} to eliminate the leftmost nonzero term of $(x_r)$ and $(y_r)$ and obtain sequences $(x'_r)$ and $(y'_r)$ which are zero everywhere except for $r=-k+1,\ldots,-2,-1,0$. This means
\begin{align*}
(x'_r)\cdot(F_r) & =x'_{-k+1}F_{-k+1}+\ldots+x'_{-2}F_{-2}+x'_{-1}F_{-1}+x'_0F_0 \\
\intertext{and}
(y'_r)\cdot(F_r) & =y'_{-k+1}F_{-k+1}+\ldots+y'_{-2}F_{-2}+y'_{-1}F_{-1}+y'_0F_0.
\end{align*}
are both equal to $m$. Since $F_{-k+1},\ldots,F_{-1}$ span $M$ of rank $k-1$, we must have $x'_r=y'_r$ for $r=-k+1,\ldots,-2,-1$. If $x'_0=y'_0$, then $(x_r)\sim_k(x'_r)=(y'_r)\sim_k(y_r)$, and so $(x_r)$ and $(y_r)$ must be the same from Corollary \ref{cor:unique}. Otherwise, we assume without loss of generality that $x'_0>y'_0$. We now have
\begin{align*}
(x_r)\cdot(\lambda_k^r)-(y_r)\cdot(\lambda_k^r) & =(x'_r)\cdot(\lambda_k^r)-(y'_r)\cdot(\lambda_k^r) \\
& =(x'_0-y'_0)\lambda_k^0 \\
& \geq 1.
\end{align*}
Since $(y_r)$ is binary, it follows that $(x_r)\cdot(\lambda_k^r)\geq1$. However, we see from Lemma \ref{lem:bound} that
\[(x_r)\cdot(\lambda_k^r)<\sum_{\substack{r<0 \\ k\nmid -r}}{\lambda_k^r}=\lambda_k^0=1,\]
which is a contradiction. This completes the proof of the theorem.
\end{proof}

Obviously, the generalized Zeckendorf theorem over negaFibonacci sequence \cite{B} serves as an instance of Theorem \ref{th:unique}. We give another example using Gaussian integers.

\begin{example}\label{ex:unique}
Consider a Fibonacci sequence
\[\begin{array}{rrrrrrrrrr}
\ldots, & -4+4i, & 5+i, & -2-3i, & -1+2i, & 2, & -1-i, & i, & 1, & \\
& 0, & 1+i, & 2+i, & 3+2i, & 6+4i, & 11+7i, & 20+13i, & 37+24i, & \ldots.
\end{array}\]
from Example \ref{ex:basic} and \ref{ex:represent}. One can see that every representation given in Example \ref{ex:represent} only involve terms from the sequence 
\[\begin{array}{rrrrrrrrr}
\ldots, & -4+4i, & 5+i, & -2-3i, & -1+2i, & 2, & -1-i, & i, & 1.
\end{array}\]
In fact, every Gaussian integer can be uniquely written as a sum of elements from this sequence with no 3 consecutive terms. This property will be exploited when we develop Fibonacci coding for Gaussian integers in Example \ref{ex:code} in the next section.
\end{example}

\section{Multidimensional Fibonacci Coding}\label{sec:code}

%We first give a quick overview of applications of Zeckendorf's theorem in data storage and transmission.

The classical Fibonacci code of order 2 is quite simple. By virtue of Zeckendorf's theorem, integers are written in base Fibonacci and encoded as a string of 0's and 1's in reverse order together with a suffix 1. For example, 11 is encoded as 001011 since it is the sum of the third and the fifth Fibonacci numbers, which are 3 and 8 respectively. The widely-accepted Fibonacci code of higher order is not so straightforward. Note that simply adding a consecutive runs of 1 no longer makes the code uniquely decodable \cite{AF}. The integers are instead mapped to a lexicographically ordered string of 0's and 1's with a single $k$ consecutive 1's at the end. This encoding, though artificial, provides dense codes. \\

We see from Theorem \ref{th:unique} that it is possible to write any element of a module as a sum of entries from a sequence using no $k$ consecutive terms. This suggests Fibonacci coding for modules with the use of $k$ consecutive 1's as a separator. This intuition formalizes into Algorithm \ref{algo:encode}. Here, we denote $k$ consecutive 1's, $\underbrace{11\ldots1}_{k}$, by $1_k$.

\begin{algorithm}
\caption{Fibonacci Encoding}
\textbf{Input:} A Fibonacci sequence $\ldots,F_{-3},F_{-2},F_{-1}$ of order $k$ where $F_{-k}+\ldots+F_{-1}=0$ and $F_{-k+1},\ldots,F_{-1}$ span a free $\mathbb{Z}$-module $M$ of rank $k-1$, and an element $m\in M$. \\
\textbf{Output:} Fibonacci code for $m$.
\begin{algorithmic}[1]
\STATE If $m=0$, output $1_k$. END.
\STATE Write $m$ as
\[x_{-k+1}F_{-k+1}+\ldots+x_{-2}F_{-2}+x_{-1}F_{-1}.\]
Set $x_0=0$ and $x_r=0$ for $r\leq-k$.
\STATE If $\check{x}=\min\{x_{-k+1},\ldots,x_{-2},x_{-1}\}<0$, subtract $\check{x}$ from $x_{-k},x_{-k+1},\ldots,x_{-2},x_{-1}$.
\STATE Find the largest index $s<0$ such that $x_{s-k+1},\ldots,x_{s}\geq1$. If there is none, go to Step 5. Otherwise, subtract 1 from $x_{s-k+1},\ldots,x_{s}$, add 1 to $x_{s+1}$, and repeat this step.
\STATE Find the largest index $s<0$ such that $x_s\geq2$. If there is none, go to Step 6. Otherwise, subtract 2 from $x_s$, add 1 to $x_{s-k}$ and $x_{s+1}$, and go to Step 4.
\STATE Find the smallest index $s<0$ such that $x_s=1$. Output $x_{-1}x_{-2}x_{-3}\ldots x_{s+2}x_{s+1}01_k$. END.
\end{algorithmic}
\label{algo:encode}
\end{algorithm}

Note that the coefficients are encoded in reverse order, and we replace the last 1 with $01_k$ to make our code uniquely decodable. This frees us from having to keep the lookup table, but makes the code suboptimal in terms of density.
%each codeword contains exactly one instance of $1_k$ at the end. Thus, the suffix $1_k$ acts as a separator for our code.
In addition, note that the algorithm simply mimics the arguments given in the proof of Theorem \ref{th:unique}. This guarantees that the algorithm terminates since the operations performed in Step 4 and 5 increase the lexicographical order of the sequence. At every step, the value $\sum_{r\in\mathbb{Z}^-}{x_rF_r}$ remains unchanged. This fact makes the decoding, which is outlined as Algorithm \ref{algo:decode}, straightforward.

\begin{algorithm}
\caption{Fibonacci Decoding}
\textbf{Input:} A Fibonacci sequence $\ldots,F_{-3},F_{-2},F_{-1}$ of order $k$, and a binary string $x_1x_2x_3\ldots x_{s-1}x_s1_k$. \\
\textbf{Output:} An element $m$.
\begin{algorithmic}[1]
\STATE If $x_1x_2x_3\ldots x_{s-1}x_s$ is empty, output $0$. END.
\STATE Output $x_{1}F_{-1}+x_{2}F_{-2}+\ldots+x_{s-1}F_{-s+1}+F_{-s}$. END.
\end{algorithmic}
\label{algo:decode}
\end{algorithm}

To illustrate, we consider once again the Gaussian Fibonacci sequence used through Examples \ref{ex:basic}, \ref{ex:represent}, and \ref{ex:unique}.

\begin{example}\label{ex:code}
Consider the Fibonacci sequence
\[\begin{array}{rrrrrrrrr}
\ldots, & -4+4i, & 5+i, & -2-3i, & -1+2i, & 2, & -1-i, & i, & 1,
\end{array}\]
and let $m=-2+3i$. To encode, we initialize $x_{-2}=3$, $x_{-1}=-2$ and obtain $x_{-3}=2$, $x_{-2}=5$, $x_{-1}=0$ after Step 3. We shorthand this sequence as $2,5,0$. Iterations of Step 4 and 5 now yield
\[2,5,0\rightarrow1,0,2,3,1\rightarrow1,0,1,2,0\rightarrow2,0,1,0,1\rightarrow1,0,0,0,1,1,0,1.\]
Thus, $m=-2+3i$ is encoded as $10110000111$. This string can then be decoded as $F_{-1}+F_{-3}+F_{-4}+F_{-8}=1+(-1-i)+2+(-4+4i)=-2+3i$. Table \ref{table:gauss} below gives the Fibonacci code under this sequence for some Gaussian integers. The length of the encoded Gaussian integer $a+bi$ where $-10\leq a,b\leq10$ is illustrated in Figure \ref{fig:length1}. This is extended to $-1000\leq a,b\leq1000$ in Figure \ref{fig:length2}. Here, region of the same color represents Gaussian integer of the same encoded length using basis given in this example.

\begin{table}[H] \centering
\caption{Fibonacci code for some Gaussian integers}
\label{table:gauss}
\begin{tabular}{|c|c|c|c|c|c|}
\hline
& $a=-2$ & $a=-1$ & $a=0$ & $a=1$ & $a=2$ \\
\hline
$b=2i$ & $01100111$ & $00000111$ & $10000111$ & $00010111$ & $1001011$ \\
\hline
$b=i$ & $00100111$ & $10100111$ & $00111$ & $10111$ & $0100111$ \\
\hline
$b=0$ & $110010111$ & $010111$ & $111$ & $0111$ & $0000111$ \\
\hline
$b=-i$ & $100010111$ & $000111$ & $100111$ & $0010111$ & $1010111$ \\
\hline
$b=-2i$ & $010000111$ & $110000111$ & $010100111$ & $110000111$ & $0110010111$ \\
\hline
\end{tabular}
\end{table}

\begin{figure} \centering
\caption{The encoded length of $a+bi$ where $-10\leq a,b\leq10$}
\label{fig:length1}
\includegraphics[width=8cm]{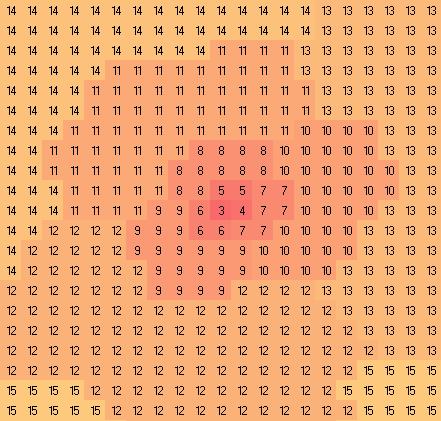}
\end{figure}

\begin{figure} \centering
\caption{The encoded length of $a+bi$ where $-1000\leq a,b\leq1000$}
\label{fig:length2}
\includegraphics[width=10cm]{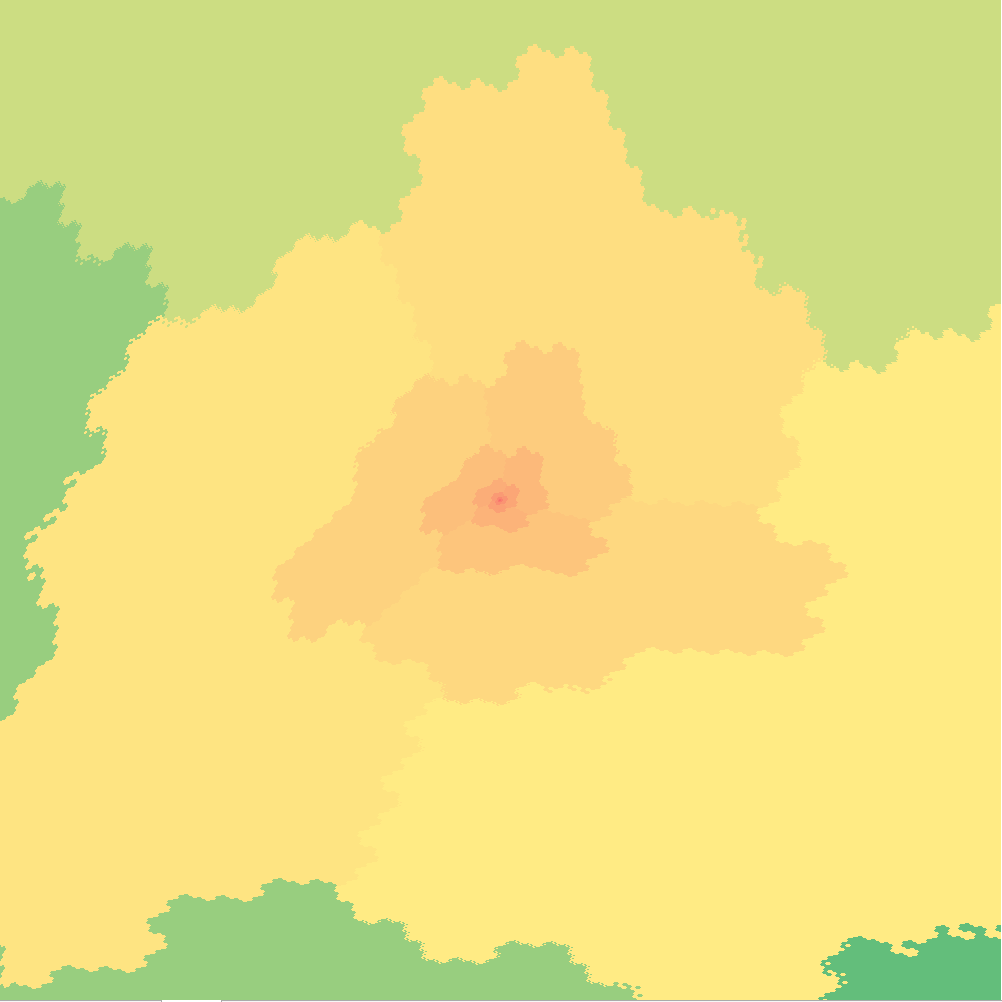}
\end{figure}

\end{example}

Clearly, one may replace $i$ in Example \ref{ex:code} by any other quadratic integer and obtain a Fibonacci coding for the corresponding ring of quadratic integers. Next, we give an example of Fibonacci coding for a lattice.

\begin{example}
Consider the lattice $E_8$ given by
\[E_8=\left\{\mathbf{x}\in\mathbb{Z}^8\cup(\mathbb{Z}+\tfrac{1}{2})^8\mid\textstyle{\sum{x_i}}\equiv0\pmod{2}\right\}.\]
This lattice provides optimal sphere packing and kissing number in 8 dimensions and has many other interesting properties \cite{CS}. We set the basis for $E_8$ as
\begin{align*}
\mathbf{v}_1 & =(2,0,0,0,0,0,0,0) & \mathbf{v}_2 & =(-1,1,0,0,0,0,0,0) \\
\mathbf{v}_3 & =(0,-1,1,0,0,0,0,0) & \mathbf{v}_4 & =(0,0,-1,1,0,0,0,0) \\
\mathbf{v}_5 & =(0,0,0,-1,1,0,0,0) & \mathbf{v}_6 & =(0,0,0,0,-1,1,0,0) \\
\mathbf{v}_7 & =(0,0,0,0,0,-1,1,0) & \mathbf{v}_8 & =(\tfrac{1}{2},\tfrac{1}{2},\tfrac{1}{2},\tfrac{1}{2},\tfrac{1}{2},\tfrac{1}{2},\tfrac{1}{2},\tfrac{1}{2})
\end{align*}
and consider a Fibonacci sequence of order 9 given by
\[\ldots,2\mathbf{v}_2-\mathbf{v}_1,2\mathbf{v}_1,-\mathbf{v}_1-\mathbf{v}_2-\mathbf{v}_3-\mathbf{v}_4-\mathbf{v}_5-\mathbf{v}_6-\mathbf{v}_7-\mathbf{v}_8,\mathbf{v}_8,\mathbf{v}_7,\mathbf{v}_6,\mathbf{v}_5,\mathbf{v}_4,\mathbf{v}_3,\mathbf{v}_2,\mathbf{v}_1.\]
It follows from Theorem \ref{th:unique} that every element of $E_8$ can be written as a sum of elements from this sequence with no 9 consecutive terms. For an example, let $\mathbf{m}=(\tfrac{1}{2},1\tfrac{1}{2},1\tfrac{1}{2},\tfrac{1}{2},\tfrac{1}{2},\tfrac{1}{2},\tfrac{1}{2},\tfrac{1}{2})\in E_8$. Then, $\mathbf{m}=\mathbf{v}_1+2\mathbf{v}_2+\mathbf{v}_3+\mathbf{v}_8=(2\mathbf{v}_2-\mathbf{v}_1)+(2\mathbf{v}_1)+(\mathbf{v}_8)+(\mathbf{v}_3)$ and can be encoded as $00100001010111111111$.
\end{example}

\subsection{Compression Efficiency}

In this subsection, we give an overview for the compression efficiency of the proposed scheme, both theoretically and numerically. We refer interested readers to \cite{S} for a comprehensive survey on the topic. \\

Let $\lambda_k$ be the dominating root of $x^k-x^{k-1}-x^{k-2}-\ldots-x-1=0$, i.e., the $k^{\mathrm{th}}$-order golden ratio. It takes approximately
\[\log_{\lambda_k}(A)+k-1\]
bits to encode a single integer $A$ using the Fibonacci code of order $k$. Unfortunately, due to the erratic nature of higher-order number system, a simple formula cannot be made to approximate the number of bits needed to encode the string $A_1,A_2,\ldots,A_k$ using our proposed algorithm (see also Figure \ref{fig:length2}). Here, we shall attempt to give a very crude estimate. Theorem \ref{th:unique} guarantees a one-to-one correspondence between modules element and sequences of 0's and 1's with no $k$ consecutive 1's. Consider a $k$-dimensional hypercube with sides parallel to the axes, and the origin and $(A_1,A_2,\ldots,A_k)$ are opposite vertices. It takes at least $\log_{\lambda_k}(A_1A_2\ldots A_k)$ bits to assign each integer point in this hypercube a unique binary strings under Fibonacci coding. Thus, if $A_1,A_2,\ldots,A_k$ are sufficiently large, one may estimate the number of bits needed to encode them altogether as
\[\log_{\lambda_k}(A_1A_2\ldots A_k)+k.\]
If $A$ is the geometric mean of $A_1,A_2,\ldots,A_k$, then the per-element average of the number of bits used to write each of the $A_i$'s is
\[\log_{\lambda_k}(A)+1.\]
This can be interpreted that the burden of the $k$-bit suffix is shared across the string entries. \\

While there exist several modern compression techniques \cite{AZ,BINP,MNZB,YTH}, each with different strengths and tradeoffs, we compare multidimensional Fibonacci code with the classical one, using the well-established Huffman code as a benchmark.

We consider string of alphabets from the set $\{1,2,\ldots,n\}$. The per-element average codeword lengths when alphabets are from a uniform distribution are given in Table \ref{table:compare1}. Here, multidimensional Fibonacci codes outperform their classical counterpart in terms of compression across all range and order. This is not surprising since the burden of the suffix is shared across several alphabets in the proposed scheme.

%To compare the length of classical Fibonacci codes and multidimensional Fibonacci codes, we encode a string of random numbers from varying range and compute the per-element average. The results

\begin{table}[H] \centering
\caption{Average codeword lengths under uniform distribution}
\label{table:compare1}
\begin{tabular}{|c|c|c|c|c|}
\hline
$n$ & 128 & 256 & 512 & 1024 \\
\hline
\hline
Huffman (= binary encoding) & 7 & 8 & 9 & 10 \\
\hline
\hline
Fibonacci order 2 & 10.609 & 12.054 & 13.525 & 14.881 \\
\hline
Fibonacci order 3 & 10.485 & 11.61 & 12.72 & 13.827 \\
\hline
Fibonacci order 4 & 11.172 & 12.242 & 13.307 & 14.372 \\
\hline
\hline
Multidimensional Fibonacci order 3 & 10.384 & 11.009 & 12.5 & 13.574 \\
\hline
Multidimensional Fibonacci order 4 & 9.828 & 11.555 & 12.224 & 13.310 \\
\hline
\end{tabular}
\end{table}

Next, Table \ref{table:compare2} consider the case when alphabets are drawn from Zipf distribution, and also when Fibonacci codes are used to encode bigrams from the set of alphabets of size 32. Here, the performance of multidimensional Fibonacci codes are slightly behind their classical counterpart. This is probably due to the fact that multidimensional encoder gets a hold-back when ``frequent'' and ``infrequent'' alphabets are in the same group. Finally, we see that the performance of multidimensional codes are comparable to the classical ones when applied to bigrams.

%We next consider a string of texts from the set $\{1,2,\ldots,256\}$ under Zipf and Geo-$\lambda_3$. Here, Zipf is a distribution that mimics that of natural language, and Geo-$\lambda_3$ is a geometric distribution with parameter $\lambda_3$, which is a natural and optimal distribution for Fibonacci codes of order 3. Also, we consider bigrams of a string of texts from $\{1,2,\ldots,16\}$ under Geo-$\lambda_3$, practically constitutes 256 raw characters.

\begin{table}[H] \centering
\caption{Average codeword lengths under Zipf distribution}
\label{table:compare2}
\begin{tabular}{|c|c|c|c|c|c|}
\hline
$n$ & 128 & 256 & 512 & 1024 & $32\times32$ bigram \\
\hline
\hline
Huffman & 5.598 & 6.258 & 6.901 & 7.537 & 8.33 \\
\hline
\hline
Fibonacci order 2 & 5.92 & 6.604 & 7.299 & 7.991 & 9.223 \\
\hline
Fibonacci order 3 & 6.54 & 7.104 & 7.667 & 8.23 & 9.273 \\
\hline
Fibonacci order 4 & 7.449 & 7.985 & 8.519 & 9.052 & 10.046 \\
\hline
\hline
Multidimensional Fibonacci order 3 & 6.617 & 7.203 & 7.807 & 8.409 & 9.494 \\
\hline
Multidimensional Fibonacci order 4 & 7.471 & 8.015 & 8.558 & 9.1 & 10.106 \\
\hline
\end{tabular}
\end{table}

We remark here that while Huffman code outperform Fibonacci code numerically, it works well only when there are finite alphabets and the distribution is known a priori. Fibonacci code is also advantageous in that it provides robustness against errors and works naturally on the set of integers. In this sense, multidimensional Fibonacci code could be of particular interest when integers of unknown sign and magnitude are to be encoded.

\remove{Figure \ref{fig:compare} below shows the length of multidimensional Fibonacci codes of various order when a string of integer $n$ is encoded.

\begin{figure}[H] \centering
\caption{Compare}
\label{fig:compare}
\includegraphics[width=10cm]{Compare}
\end{figure}}

\section{Conclusion}\label{sec:conclusion}

In this paper, we generalize Zeckendorf's theorem to modules. The notion of equivalent sequences allows us to identify elements that can be represented as a sum of entries from a Fibonacci sequence of order $k$. This results in the necessary and sufficient conditions for a Fibonacci sequence of higher order to generate a module. In addition, under certain circumstances the representation is unique, allowing us to establish Fibonacci coding for modules. Future work involves identifying other conditions to which the representation remains unique. Under such environments, one can view Zeckendorf representation as a number system and develop generalized Zeckendorf arithmetic. It would also be interesting to study the proposed coding algorithm from the perspective of data compression and computational complexity.

\section*{Acknowledgments}
The authors wish to thank Eaksit Buathong-iem for his assistance. This work is supported by the Thailand Research Fund under Grant MRG6180192 and the Centre of Excellence in Mathematics, the Commission on Higher Education, Thailand.


\begin{thebibliography}{99}
\bibitem{AZ}
A. V. Anisimov and I. O. Zavadskyi, Variable-Length Prefix Codes With Multiple Delimiters, IEEE Transactions on Information Theory \textbf{63} (2017), no. 5, 2885--2895.
\bibitem{AF}
A. Apostolico, A. Fraenkel, Robust Transmission of Unbounded Strings Using Fibonacci Representations, IEEE Transactions on Information Theory \textbf{33} (1987), 238--245.
\bibitem{BaM}
I. Ben-Ari and S. J. Miller, A Probabilistic Approach to Generalized Zeckendorf Decompositions, SIAM Journal on Discrete Mathematics \textbf{30} (2016), no. 2, 1302--1332.
\bibitem{Be}
Berzsenyi, Gaussian Fibonacci Numbers, Fibonacci Quarterly \textbf{15} (1977), no. 3, 233--236.
\bibitem{BDEMMTW}
A. Best, P. Dynes, X. Edelsbrunner, B. McDonald, S. J. Miller, C. Turnage-Butterbaugh, and M. Weinstein, Gaussian Behavior of the Number of Summands in Zeckendorf Decompositions in Small Intervals, Fibonacci Quarterly \textbf{52} (2014), no. 5, 47--53.
\bibitem{BINP}
N. R. Brisaboa, E. L. Iglesias, G. Navarro, and J. R. Param\'a, An Efficient Compression Code for Text Databases, Proceedings of European Conference on IR Research (2003), 468--481.
%\bibitem{Br}
%J. L. Broun, Unique Representation of Integers as Sums of Distinct Lucas Numbers, Fibonacci Quarterly \textbf{7} (1969), no. 3, 243--252.
\bibitem{B}
M. W. Bunder, Zeckendorf Representations Using Negative Fibonacci Numbers, Fibonacci Quarterly \textbf{30} (1992), 111--115.
\bibitem{CHS}
L. Carlitz, V. E. Hoggatt, Jr., and R. Scoville, Fibonacci Representations of Higher Order, Fibonacci Quarterly \textbf{10} (1972), no. 1, 43--69.
\bibitem{C}
J. Cigler, Some Algebraic Aspects of Morse Code Sequences, Discrete Mathematics and Theoretical Computer Science \textbf{6} (2003), 55--68.
\bibitem{CS}
J. Conway, and N.J.A. Sloane, Sphere Packings, Lattices and Groups, Springer.
\bibitem{DDKMMV}
P. Demontigny, T. Do, A. Kulkarni, S. J. Miller, D. Moon, and U. Varma, Generalizing Zeckendorf's
Theorem to $f$-decompositions, Journal of Number Theory \textbf{141} (2014), 136--158.
\bibitem{E}
P. Elias, Universal Codeword Sets and Representations of the Integers, IEEE Transactions on Information Theory \textbf{21} (1975), 194--203.
\bibitem{ER}
S. Even and M. Rodeh, Economical Encoding of Commas Between Strings, Communications of the ACM \textbf{21} (1978), 315--317.
%\bibitem{F1}
%G. D. Forney, Coset Codes--Part I: Introduction and Geometrical Classification, IEEE Transactions on Information Theory \textbf{34} (1988), no. 5, 1123--1151.
\bibitem{FK}
A. Fraenkel and S. Klein, Robust Universal Complete Codes As Alternatives to Huffman Codes, Technical Report CS85-16, The Weizmann Institute of Science, Rehovot, 1985.
\bibitem{FK96}
A. Fraenkel and S. Klein, Robust Universal Complete Codes for Transmission and Compression, Discrete Applied Mathematics \textbf{64} (1996), 31--55.
\bibitem{G}
I. J. Good, Complex Fibonacci and Lucas Numbers, Continued Fractions, and the Square Root of the Golden Ratio, Journal of the Operational Research Society \textbf{43} (1992), no. 8, 837--842.
\bibitem{GTNP}
P. J. Grabner, R. F. Tichy, I. Nemes, and A. Peth\H o, Generalized Zeckendorf Expansions, Applied Mathematics Letters \textbf{7} (1994), no. 2, 25--28.
\bibitem{Hal}
E. Halsey, The Fibonacci Number $F_u$ where $u$ is not an Integer, Fibonacci Quarterly \textbf{3} (1965), no. 2, 147--152.
\bibitem{Ha}
C. J. Harman, Complex Fibonacci Numbers, Fibonacci Quarterly \textbf{19} (1981), no. 1, 82--86.
\bibitem{H}
D. Huffman, A Method for the Construction of Minimum Redundancy Codes, Proceedings of the IRE (1952), 1098--1101.
\bibitem{J}
J. H. Jordan, Gaussian Fibonacci and Lucas Numbers, Fibonacci Quarterly \textbf{3} (1965), no. 4, 315--318.
%\bibitem{KS}
%I. K\'atai, J. Szab\'o, Canonical Number Systems for Complex Integers, Acta Sci. Math., \textbf{37} (1975), 255--260.
\bibitem{KlS}
S. T. Klein and D. Shapira, Random Access to Fibonacci Encoded Files, Discrete Applied Mathematics \textbf{212} (2016), 115--128.
\bibitem{K}
D. E. Knuth, Fibonacci Multiplication, Applied Mathematics Letters \textbf{1} (1988), no. 1, 57--60.
\bibitem{L}
C. G. Lekkerkerker, Voorstelling van natuurlyke getallen door een som van getallen van Fibonacci,
Simon Stevin \textbf{29} (1951-1952), 190--195.
%T. Lengyel, A Counting Based Proof of the Generalized Zeckendorf's Theorem, Fibonacci Quarterly 44 (4), 324--325, 2006.
\bibitem{M}
E. P. Miles, Jr., Generalized Fibonacci Numbers and Associated Matrices, American Math. Monthly, \textbf{67} (1960), 745--57.
\bibitem{MW}
S. J. Miller and Y. Wang, From Fibonacci Numbers to Central Limit Type Theorems, Journal of Combinatorial Theory, Series A \textbf{119} (2012), no. 7, 1398--1413.
\bibitem{MNZB}
E. S. de Moura, G. Navarro, N. Ziviani, and R. Baeza-Yates, Fast and Flexible Word Searching on Compressed text, ACM Transactions on Information Systems \textbf{18} (2000), no. 2, 113--139.
\bibitem{P}
F. D. Parker, A Fibonacci function, Fibonacci Quarterly \textbf{6} (1968), no. 1, 1--2.
%\bibitem{PHK}
%P. Pooksombat, J. Harshan, and W. Kositwattanarerk, On Shaping Complex Lattice Constellations from Multi-level Constructions, Proceedings of IEEE International Symposium on Information Theory (2017), 2598--2602.
\bibitem{S}
D. Salomon, Data Compression The Complete Reference, Third Edition, Springer--Verlag, New York, 2004.
\bibitem{Sa}
S. Sato, Fibonacci Sequence and its Generalizations Hidden in Algorithms for Generating Morse Codes, Applications of Fibonacci Numbers, Proceedings of The Fifth International Conference on Fibonacci Numbers and Their Applications (1993), 481--486.
%\bibitem{Sc}
%M. Schwartz, Mobile Wireless Communications, Cambridge University Press, 2005.
\bibitem{U}
N. Utgoff, A generalization of Zeckendorf's Theorem, preprint.
\bibitem{WKP}
J. Walder, M. Kratky, and J. Platos, Fast Fibonacci Encoding Algorithm, DATESO (2010).
\bibitem{WKBPS}
J. Walder, M. Kratky, R. Baca, J. Platos, and V. Snasel, Fast Decoding Algorithms for Variable-Lengths Codes, Information Sciences \textbf{183} (2012), 66--91.
\bibitem{YTH}
H. Yamamoto, M. Tsuchihashi, and J. Honda, Almost Instantaneous Fixed-to-Variable Length Codes, IEEE Transactions on Information Theory \textbf{61} (2015), 6432--6443.
%\bibitem{Z}
%R. Zamir, Lattice Coding for Signals and Networks, Cambridge University Press, 2014.
\end{thebibliography}
\end{document}